\documentclass[11pt]{amsart}%
\usepackage{amssymb}
\usepackage{amsfonts}
\usepackage{amsmath}
\usepackage{graphicx}
\usepackage{hyperref}%
\newtheorem{theorem}{Theorem}
\theoremstyle{plain}

\newtheorem{definition}{Definition}

\newtheorem{lemma}{Lemma}

\newtheorem{proposition}{Proposition}
\newtheorem{remark}{Remark}

\numberwithin{equation}{section}
\begin{document}
\title{Non-Archimedean Coulomb Gases}
\author{Sergii M. Torba}
\email{storba@math.cinvestav.edu.mx}

\author{W. A. Z\'{u}\~{n}iga-Galindo}
\email{wazuniga@math.cinvestav.edu.mx}
\address{Centro de Investigaci\'{o}n y de Estudios Avanzados del Instituto
Polit\'{e}cnico Nacional\\
Departamento de Matem\'{a}ticas, Unidad Quer\'{e}taro\\
Libramiento Norponiente \#2000, Fracc. Real de Juriquilla. Santiago de
Quer\'{e}taro, Qro. 76230\\
M\'{e}xico.}

\thanks{The first named author was partially supported by Conacyt Grant No. 222478. The second named author was partially supported by Conacyt Grant No. 250845.}

\begin{abstract}
This article aims to study the Coulomb gas model over the $d$-dimen\-sional
$p$-adic space. We establish the existence of equilibria measures and the
$\Gamma$-limit for the Coulomb energy functional when the number of
configurations tends to infinity. For a cloud of charged particles confined
into the unit ball, we compute the equilibrium measure and the minimum of its
Coulomb energy functional. In the $p$-adic setting the Coulomb energy is the
continuum limit of the minus a hierarchical Hamiltonian attached to a spin
glass model with a $p$-adic coupling.

\end{abstract}
\maketitle
\tableofcontents

\section{Introduction}

In this article we initiate the study of Coulomb gases on the $d$-dimensional
$p$-adic space $\mathbb{Q}_{p}^{d}$. More precisely, we give $p$-adic
counterparts of the existence and characterization of the equilibrium measure,
see Theorem \ref{Theorem_1}, the $\Gamma$-convergence of the Coulomb energy
functional, see Theorem \ref{Theo_1A}, and the convergence of the minimizers
of this functional, see Theorem \ref{Theorem_2}. For the classical
counterparts the reader may consult, for instance, Serfaty's book
\cite[Proposition 2.8, Theorems 2.1, 2.2]{Serfaty-Book}.

From a mathematical perspective, the results presented here are framed in the
probability and potential theory over ultrametric spaces. Probability over
ultrametric spaces has been studied extensively during the last thirty years,
see e.g. \cite{Bendikov1}, \cite{Bendiko t al}, \cite{D-K-K-V} and the
references therein, due, among several things, to the emergence of the use of
ultrametric spaces in physical models, see e.g. \cite{KKZuniga}, \cite{R-T-V},
\cite{V-V-Z}, \cite{Zuniga-LNM-2016} and the references therein. On the other
hand, the study of potential theory over locally compact Abelian groups, see
e.g. \cite{Berg-Forst}, and over metric spaces, see e.g. \cite{Adams et al},
\cite{Bjorn et al}, is a classical matter.

From a physical perspective, the study of models over ultrametric spaces
started in the middle 80s with the works of Frauenfelder, Parisi, Stein, among
others, see e.g. \cite{D-K-K-V}, \cite{Fraunfelder et al}, \cite{M-P-V},
\cite{R-T-V}, see also \cite{Av-4}, \cite{Av-5}, \cite{Khrennikov2},
\cite{KKZuniga}, \cite{Zuniga-LNM-2016}, and the references therein. The key
idea is that the space of states of certain physical systems have a natural
structure of ultrametric space. An ultrametric space $(M,d)$ is a metric space
$M$ with a distance satisfying the strong triangle inequality $d(A,B)\leq
\max\left\{  d\left(  A,C\right)  ,d\left(  B,C\right)  \right\}  $ for any
three points $A$, $B$, $C$ in $M$.

The Ising models over ultrametric spaces have been studied intensively, see
e.g. \cite{DysonFreeman}, \cite{Gubser et al}, \cite{Khrennikov et al},
\cite{Khrennikov2}, \cite{Lerner-Missarov}, \cite{Mis}, \cite{Mukhamedov-1},
\cite{Mukhamedov-2}, \cite{Mukhamedov-3}, \cite{Sinai} and the references
therein, \ motivated, among several things, by the hierarchical Ising model
introduced in \cite{DysonFreeman}. The hierarchical Hamiltonian introduced by
Dyson in \cite{DysonFreeman} can be naturally studied in $p$-adic spaces, see
e.g. \cite{Lerner-Missarov}, \cite{Gubser et al}. These Hamiltonians are
self-similar with respect to suitable scale groups.

A $p$-adic number is a series of the form%
\begin{equation}
x=x_{-k}p^{-k}+x_{-k+1}p^{-k+1}+\ldots+x_{0}+x_{1}p+\ldots,\text{ with }%
x_{-k}\neq0\text{,} \label{p-adic-number}%
\end{equation}
where $p$ denotes a fixed prime number, and the $x_{j}$s \ are $p$-adic
digits, i.e. numbers in the set $\left\{  0,1,\ldots,p-1\right\}  $. There are
natural field operations, sum and multiplication, on series of form
(\ref{p-adic-number}). The set of all possible $p$-adic numbers constitutes
the field of $p$-adic numbers $\mathbb{Q}_{p}$. There is also a natural norm
in $\mathbb{Q}_{p}$ defined as $\left\vert x\right\vert _{p}=p^{k}$, for a
nonzero $p$-adic number $x$ of the form (\ref{p-adic-number}). The field of
$p$-adic numbers with the distance induced by $\left\vert \cdot\right\vert
_{p}$ is a complete ultrametric space. The ultrametric property refers to the
fact that $\left\vert x-y\right\vert _{p}\leq\max\left\{  \left\vert
x-z\right\vert _{p},\left\vert z-y\right\vert _{p}\right\}  $ for any $x$,
$y$, $z$ in $\mathbb{Q}_{p}$. As a topological space, $\left(  \mathbb{Q}%
_{p},\left\vert \cdot\right\vert _{p}\right)  $ is completely disconnected,
i.e. the connected components are points. The field of $p$-adic numbers has a
fractal structure, see e.g. \cite{A-K-S}, \cite{V-V-Z}. We extend the $p$-adic
norm to $\mathbb{Q}_{p}^{d}$, by taking $\left\Vert \left(  x_{1},\ldots
,x_{d}\right)  \right\Vert =\max_{i}\left\vert x_{i}\right\vert _{p}$.

For $\alpha>0$, the $d$-dimensional $p$-adic Coulomb kernel is defined as%

\[
g_{\alpha}\left(  x\right)  =%
\begin{cases}
||x||_{p}^{\alpha-d}, & \text{ if }\alpha\neq d\\
\ln||x||_{p}, & \text{ if }\alpha=d.
\end{cases}
\]
This kernel is similar to the classic one, however, in the $p$-adic setting,
we have a family of kernels parametrized by $\alpha>0$. In this article we
consider only the kernels $\frac{1}{||x||_{p}^{d-\alpha}}$, with $d>\alpha$.
The function $g_{\alpha}\left(  x\right)  $ is the fundamental solution of a
`$p$-adic Poisson's equation.' In the $p$-adic framework, there are infinitely
many `Laplacians'. By a Laplacian we mean an operator $A$ such that the
semigroup generated by $-A$ is Markovian. We pick the simplest possible
Laplacian in dimension $d$, the Taibleson operator $\boldsymbol{D}^{\alpha}$,
$\alpha>0$, which is a pseudodifferential operator defined as $\mathcal{F}%
\left(  \boldsymbol{D}^{\alpha}\varphi\right)  =||\xi||_{p}^{\alpha
}\mathcal{F}\varphi$, where $\mathcal{F}$ denotes the Fourier transform.
Notice that here $\alpha$ is an arbitrary positive number, while in the
classical case, similar operators exist only if $\alpha\in\left(  0,2\right]
$. If we consider $g_{\alpha}\left(  x\right)  $ as distribution, then
\[
\boldsymbol{D}^{\alpha}g_{\alpha}=-C_{d,\alpha}\delta,
\]
where $\delta$ denotes the Dirac distribution at the origin and $C_{d,\alpha}%
$\ is a constant, see Section \ref{Section_Taibleson_operator}.

Let $\mathcal{P}(\mathbb{Q}_{p}^{d})$ be the space of probability measures on
$\mathbb{Q}_{p}^{d}$. The Coulomb energy of the measure $\mu$\ is defined as
\[
\mathcal{E}_{\alpha}(\mu)=\int_{\mathbb{Q}_{p}^{d}}\,\int_{\mathbb{Q}_{p}^{d}%
}g_{\alpha}\left(  x-y\right)  d\mu\left(  x\right)  d\mu\left(  y\right)
\in\left(  -\infty,+\infty\right]  .
\]
Now we introduce and admissible potential $V:\mathbb{Q}_{p}^{d}\rightarrow
\left(  -\infty,+\infty\right]  $ satisfying the standard conditions. For this
potential we consider the Coulomb energy functional%
\[
I_{\alpha}\left(  \mu\right)  =\mathcal{E}_{\alpha}(\mu)+\int_{\mathbb{Q}%
_{p}^{d}}V\left(  x\right)  d\mu\left(  x\right)  .
\]
We show the existence of a unique minimizer $\mu_{0}$ ($\min_{\mu
\in\mathcal{P}(\mathbb{Q}_{p}^{d})}\left\{  I_{\alpha}\left(  \mu\right)
\right\}  =I_{\alpha}\left(  \mu_{0}\right)  $) called the equilibrium
measure, see Theorem \ref{Theorem_1}.

Since $\left(  \mathbb{Q}_{p}^{d},\left\Vert \cdot\right\Vert _{p}\right)  $
is a Polish space, we can use classical probability techniques to establish
Theorem \ref{Theorem_1}. This result is a $p$-adic version of the Frostman
theorem, see e.g. \cite[Theorem 2.1]{Serfaty-Book}. In the case $V\equiv0$,
this result is well-known in the context of locally compact Abelian groups,
see e.g. \cite[Theorem 16.22]{Berg-Forst}.

At first sight, Theorem \ref{Theorem_1} is not very different of the classical
one. However, there are several important differences, among them, suitable
locally constant functions are admissible potentials; second, the ultrametric
topology of $\mathbb{Q}_{p}^{d},$ imposes new restrictions on the equilibria
measures; and third, operator $\boldsymbol{D}^{\alpha}$\ is non local. This
last fact makes the computation of the equilibrium measures very difficult.

Consider the potential%
\[
V(x)=
\begin{cases}
V_{0}, & \text{if } \left\Vert x\right\Vert _{p}\leq1\\
+\infty, & \text{if } \left\Vert x\right\Vert _{p}>1,
\end{cases}
\]
where $V_{0}>0$. The energy functional $I(\mu)$ corresponds to a cloud of
charged particles confined into the unit ball. In Proposition \ref{Prop_7}, we
compute the equilibrium measure $\mu_{0}$ for $I(\mu)$. \ In the classical
approach, one applies the Laplacian to an equality of the form%
\[
\int_{\mathbb{Q}_{p}^{d}}g_{\alpha}\left(  x-y\right)  d\mu_{0}\left(
x\right)  +\frac{V}{2}=C,\qquad\text{q.e. in the support of }\mu_{0},
\]
see Theorem \ref{Theorem_1}, to obtain a formula for $\mu_{0}\left(  x\right)
$ in an open set contained in the support of $\mu_{0}$. In the $p$-adic case,
this approach is not possible due to the fact that the operator
$\boldsymbol{D}^{\alpha}$\ is non local, see Section
\ref{Section_Coulomb_gas_unit_ball} and Proposition \ref{Prop_7}.

The Hamiltonian $H_{n,\alpha}(x_{1},\ldots,x_{n})$ of the Coulomb gas
corresponding to the configuration $x_{1},\ldots,\allowbreak x_{n}%
\in\mathbb{Q}_{p}^{d}$ is defined as%
\[
H_{n,\alpha}(x_{1},\ldots,x_{n})=\sum_{i\neq j}g_{\alpha}\left(  x_{i}%
-x_{j}\right)  +n\sum_{i}V(i).
\]
Under the assumptions that $V$ is continuous and bounded from below, and that
$g_{\alpha}(x)=\frac{1}{\left\Vert x\right\Vert _{p}^{d-\alpha}}$, with
$d>\alpha>0$, we show that $\frac{1}{n^{2}}H_{n,\alpha}$ \ $\Gamma$-converges
to $I_{\alpha}\left(  \mu\right)  $, see Theorem \ref{Theo_1A}, i.e.
$I_{\alpha}\left(  \mu\right)  $ is the mean-field energy of $\frac{1}{n^{2}%
}H_{n,\alpha}$.

We also consider the configurations $x_{1},\ldots,x_{n}$ minimizing the
corresponding Hamiltonians $H_{n,\alpha}$, $n\in\mathbb{N}$ and show that
$\frac{1}{n}\sum_{i=1}^{n}\delta_{x_{i}}\rightarrow\mu_{0}$ in the weak sense
of probability measures, and that $\lim_{n\rightarrow+\infty}\frac{1}{n^{2}%
}H_{n,\alpha}\left(  x_{1},\ldots,x_{n}\right)  =I(\mu_{0})$, see Theorem
\ref{Theorem_2}.

For $L\in Z$ fixed and $l\geq-L$, set $G_{l}=p^{-L}\mathbb{Z}_{p}^{d}%
/p^{l}\mathbb{Z}_{p}^{d}$, where $\mathbb{Z}_{p}^{d}$ is the $d$-dimensional
unit ball. Then $G_{l}$ is naturally a finite ultrametric space. Consider the
Hamiltonian%
\[
H_{L,l}:=-\sum_{\widetilde{x},\,\widetilde{y}\in G_{l}}
p^{-2ld}J_{\widetilde{x}\,\widetilde{y}}\rho\left(  \widetilde
{x}\right)  \rho\left(  \widetilde{y}\right)  -
\sum_{\widetilde{x}\in G_{l}}
p^{-ld}\rho\left(  \widetilde{x}\right)  V_{0}\left(  \widetilde{x}\right)  ,
\]
where $\rho$ and $V_{0}$ are real-valued functions and the coupling
$J_{\widetilde{x}\,\widetilde{y}}$ is given by%
\[
J_{\widetilde{x}\,\widetilde{y}}=
\begin{cases}
\left\Vert \widetilde{x}-\widetilde{y}\right\Vert _{p}^{\alpha-d}, & \text{if } \widetilde{x}\neq\widetilde{y},\\
\frac{p^{-l\left(  d+\alpha\right)  }\left(  1-p^{-d}\right)  }{1-p^{-\alpha}},
& \text{if } \widetilde{x}=\widetilde{y}.
\end{cases}
\]
Then $H_{L,l}$ is the Hamiltonian of a spin glass model with $p$-adic
coupling, see \cite[Section C]{Gubser et al} . Under general conditions about
functions $\rho$ and $V_{0}$, $\ $we obtain that $-\lim_{l\rightarrow\infty
}H_{L},_{l}$ agrees with the Coulomb energy attached to the measure $\rho\, dx$ and
a potential which is infinite outside the ball $p^{-L}\mathbb{Z}_{p}^{d}$, and
that agrees with the function $V_{0}$ inside the ball $p^{-L}\mathbb{Z}_{p}^{d}$,
see Section \ref{Section_Spin_glass}.

The Coulomb gas model is related with several relevant matters, among them,
random matrices and the obstacle problem, see e.g. \cite[Chapter
2]{Serfaty-Book}. The theory of $p$-adic random matrices is not fully
developed, but it is connected with relevant number-theoretic matters, see
e.g. \cite{Ellenberg et al}, see also \cite{Evans}. We expect that the
$p$-adic Coulomb gas model will be useful in the study of $p$-adic random
matrices. On the other hand, discrete versions of the obstacle problem play a
central role in the study of sandpile models, see e.g. \cite{Levine}. Sandpile
models have been studied on infinite trees, see e.g. \cite{Maes et al}, which
are ultrametric spaces. We expect that $p$-adic versions of the obstacle
problem will play a central role in the construction of $p$-adic counterparts
of sandpile models. Finally, all the results presented in this work are still
valid if we replace $\mathbb{Q}_{p}$ by $\mathbb{F}_{p}((t))$, the field of
formal Laurent power series with coefficients in the finite field
$\mathbb{F}_{p}$ with $p$ elements. In the recent preprint \cite{Sinclair},
Sinclair and Vaaler study the partition function for a $p$-adic Coulomb gas
confined into the unit ball in the case $g_{d}\left(  x\right)  =\ln\left\Vert
x\right\Vert _{p}$. This partition function is a local zeta function attached
to the Vandermonde determinant.

\section{\label{Section_2}Basic aspects of the $p$\textbf{-}adic analysis}

In this section we collect some basic results about $p$-adic analysis that
will be used in the article. For an in-depth review of the $p$-adic analysis
the reader may consult \cite{A-K-S}, \cite{Taibleson}, \cite{V-V-Z}.

\subsection{The field of $p$-adic numbers}

Along this article $p$ will denote a prime number. The field of $p-$adic
numbers $\mathbb{Q}_{p}$ is defined as the completion of the field of rational
numbers $\mathbb{Q}$ with respect to the $p-$adic norm $|\cdot|_{p}$, which is
defined as
\[
\left\vert x\right\vert _{p}=
\begin{cases}
0, & \text{if } x=0\\
p^{-\gamma}, & \text{if } x=p^{\gamma}\frac{a}{b},
\end{cases}
\]
where $a$ and $b$ are integers coprime with $p$. The integer $\gamma=:
\operatorname{ord}(x) $, with $\operatorname{ord}(0):=+\infty$, is called
the\textit{\ }$p-$\textit{adic order of} $x$.

Any $p-$adic number $x\neq0$ has a unique expansion of the form
\[
x=p^{\operatorname{ord}(x)}\sum_{j=0}^{\infty}x_{j}p^{j},
\]
where $x_{j}\in\{0,\dots,p-1\}$ and $x_{0}\neq0$. By using this expansion, we
define \textit{the fractional part of }$x\in\mathbb{Q}_{p}$, denoted
$\{x\}_{p}$, as the rational number
\[
\left\{  x\right\}  _{p}=
\begin{cases}
0, & \text{if } x=0\text{ or }\operatorname{ord}(x)\geq0\\
p^{\operatorname{ord}(x)}\sum_{j=0}^{-\operatorname{ord}(x)-1}x_{j}p^{j}, &
\text{if } \operatorname{ord}(x)<0.
\end{cases}
\]
In addition, any non-zero $p-$adic number can be represented uniquely as
$x=p^{\operatorname{ord}(x)}\operatorname{ac}\left(  x\right)  $ where
$\operatorname{ac}\left(  x\right)  =\sum_{j=0}^{\infty}x_{j}p^{j}$,
$x_{0}\neq0$, is called the \textit{angular component} of $x$. Notice that
$\left\vert \operatorname{ac}\left(  x\right)  \right\vert _{p}=1$.

We extend the $p-$adic norm to $\mathbb{Q}_{p}^{d}$ by taking
\[
||x||_{p}:=\max_{1\leq i\leq d}|x_{i}|_{p},\qquad\text{for }x=(x_{1}%
,\dots,x_{d})\in\mathbb{Q}_{p}^{d}.
\]
We define $\operatorname{ord}(x)=\min_{1\leq i\leq d}\{\operatorname{ord}%
(x_{i})\}$, then $||x||_{p}=p^{-\operatorname{ord}(x)}$. The metric space
$\left(  \mathbb{Q}_{p}^{d},||\cdot||_{p}\right)  $ is a separable complete
ultrametric space. For $r\in\mathbb{Z}$, denote by $B_{r}^{d}(a)=\{x\in
\mathbb{Q}_{p}^{d};||x-a||_{p}\leq p^{r}\}$ \textit{the ball of radius }%
$p^{r}$ \textit{with center at} $a=(a_{1},\dots,a_{d})\in\mathbb{Q}_{p}^{d}$,
and take $B_{r}^{d}:=B_{r}^{d}(0)$. Note that $B_{r}^{d}(a)=B_{r}(a_{1}%
)\times\cdots\times B_{r}(a_{d})$, where $B_{r}(a_{i}):=\{x\in\mathbb{Q}%
_{p};|x_{i}-a_{i}|_{p}\leq p^{r}\}$ is the one-dimensional ball of radius
$p^{r}$ with center at $a_{i}\in\mathbb{Q}_{p}$. The ball $B_{0}^{d}$ equals
to the product of $d$ copies of $B_{0}=\mathbb{Z}_{p}$, \textit{the ring of
}$p-$\textit{adic integers of }$\mathbb{Q}_{p}$. We also denote by $S_{r}%
^{d}(a)=\{x\in\mathbb{Q}_{p}^{d};||x-a||_{p}=p^{r}\}$ \textit{the sphere of
radius }$p^{r}$ \textit{with center at} $a=(a_{1},\dots,a_{d})\in
\mathbb{Q}_{p}^{d}$, and take $S_{r}^{d}:=S_{r}^{d}(0)$. We notice that
$S_{0}^{1}=\mathbb{Z}_{p}^{\times}$ (the group of units of $\mathbb{Z}_{p}$),
but $\left(  \mathbb{Z}_{p}^{\times}\right)  ^{d}\subsetneq S_{0}^{d}$. The
balls and spheres are both open and closed subsets in $%
\mathbb{Q}
_{p}^{d}$. In addition, two balls in $%
\mathbb{Q}
_{p}^{d}$ are either disjoint or one is contained in the other.

As a topological space $\left(
\mathbb{Q}
_{p}^{d},||\cdot||_{p}\right)  $ is totally disconnected, i.e. the only
connected \ subsets of $%
\mathbb{Q}
_{p}^{d}$ are the empty set and the points. A subset of $%
\mathbb{Q}
_{p}^{d}$ is compact if and only if it is closed and bounded in $%
\mathbb{Q}
_{p}^{d}$, see e.g. \cite[Section 1.3]{V-V-Z}, or \cite[Section 1.8]{A-K-S}.
The balls and spheres are compact subsets. Thus $\left(
\mathbb{Q}
_{p}^{d},||\cdot||_{p}\right)  $ is a locally compact topological space.

We will use $\Omega\left(  p^{-r}||x-a||_{p}\right)  $ to denote the
characteristic function of the ball $B_{r}^{d}(a)$. We will use the notation
$1_{A}$ for the characteristic function of a set $A$. Along the article $dx$
will denote a Haar measure on $\left(
\mathbb{Q}
_{p}^{d},+\right)  $ normalized so that $\int_{%
\mathbb{Z}
_{p}^{d}}dx=1.$

\subsection{Some function spaces}

A complex-valued function $\varphi$ defined on $\mathbb{Q} _{p}^{d}$ is called
\textit{locally constant} if for any $x\in\mathbb{Q} _{p}^{d}$ there exist an
integer $l(x)\in\mathbb{Z}$ such that
\begin{equation}
\varphi(x+x^{\prime})=\varphi(x)\qquad\text{for }x^{\prime}\in B_{l(x)}^{d}.
\label{Eq_L_con}%
\end{equation}
A function $\varphi:\mathbb{Q}_{p}^{d}\rightarrow\mathbb{C}$ is called a
\textit{Bruhat-Schwartz function,} or a \textit{test function,} if it is
locally constant with compact support. In this case, there exists
$l\in\mathbb{Z}$, independent of $x$, such that (\ref{Eq_L_con}) holds. The
largest of such numbers $l=l\left(  \varphi\right)  $ is called the
\textit{index of local constancy} of $\varphi$. The $\mathbb{C}$-vector space
of Bruhat-Schwartz functions is denoted by $\mathcal{D}:=\mathcal{D}%
(\mathbb{Q}_{p}^{d})$. We will denote by $\mathcal{D}_{\mathbb{R}%
}:=\mathcal{D} _{\mathbb{R}}(\mathbb{Q} _{p}^{d})$, the $\mathbb{R}$-vector
space of test functions. The convergence in $\mathcal{D}$ is defined in the
following way: $\varphi_{k}\rightarrow0$, $k\rightarrow\infty$, in
$\mathcal{D}$ if and only if

\begin{itemize}
\item[(i)] all the $\varphi_{k}$s are supported in a ball $B_{N}^{d}$ and have
indices of local constancy $l(\varphi_{k})\geq l$, with $N$ and $l$
independent of $k$;

\item[(ii)] $\varphi_{k}\rightarrow0$ uniformly in $\mathbb{Q}_{p}^{d}$.
\end{itemize}

Let $\mathcal{D}^{\prime}:=\mathcal{D}^{\prime}(\mathbb{Q} _{p}^{d})$ denote
the set of all continuous functionals (distributions) on $\mathcal{D}$. We
will denote by $\mathcal{D}_{\mathbb{R}}^{\prime}:=\mathcal{D}_{\mathbb{R}%
}^{\prime}(\mathbb{Q} _{p}^{d})$ the $\mathbb{R}$-vector space of
distributions. The convergence in $\mathcal{D}^{\prime}$ is the weak
convergence: $T_{k}\rightarrow0$, $k\rightarrow\infty$, in $\mathcal{D}%
^{\prime}$ if $\left(  T_{k},\varphi\right)  \rightarrow0$, $k\rightarrow
\infty$, for any $\varphi\in\mathcal{D}$.

Given $\rho\in\lbrack1,\infty)$ and an open subset $U\subset\mathbb{Q}_{p}%
^{d}$, we denote by $L^{\rho}:=L^{\rho}\left(  U\right)  $ the $\mathbb{C}%
$-vector space of all the complex valued functions $g$ defined on
$U$\ satisfying $\left\Vert g\right\Vert _{\rho}=\left\{  \int_{U}\left\vert
g\left(  x\right)  \right\vert ^{\rho}dx\right\}  ^{\frac{1}{\rho}}<\infty$,
and $L^{\infty}\allowbreak:=L^{\infty}\left(  U\right)  $ denotes the
$\mathbb{C}-$vector space of all the complex valued functions $g$ defined in
$U$ such that the essential supremum of $|g|$ is bounded. The corresponding
$\mathbb{R}$-vector spaces are denoted as $L_{\mathbb{R}}^{\rho}%
\allowbreak:=L_{\mathbb{R}}^{\rho}\left(  U\right)  $, $1\leq\rho\leq\infty$.

Let $U$ be an open subset of $\mathbb{Q} _{p}^{d}$, we denote by
$\mathcal{D}(U)$ the $\mathbb{C}$-vector space of all test functions from
$\mathcal{D}(\mathbb{Q} _{p}^{d})$ with supports in $U$. For each $\rho
\in\lbrack1,\infty)$, $\mathcal{D}(U)$ is dense in $L^{\rho}\left(  U\right)
$, see e.g. \cite[Proposition 4.3.3]{A-K-S}.

\subsection{Fourier transform}

Set $\chi_{p}(y):=\exp(2\pi i\{y\}_{p})$ for $y\in\mathbb{Q} _{p}$. The map
$\chi_{p}(\cdot)$ is an additive character on $\mathbb{Q}_{p}$, i.e. a
continuous map from $\left(  \mathbb{Q}_{p},+\right)  $ into $S$ (the unit
circle considered as multiplicative group) satisfying $\chi_{p}(x_{0}%
+x_{1})=\chi_{p}(x_{0})\chi_{p}(x_{1})$, $x_{0},x_{1}\in\mathbb{Q}_{p}$. The
additive characters of $\mathbb{Q} _{p}$ form an Abelian group which is
isomorphic to $\left(  \mathbb{Q}_{p},+\right)  $, the isomorphism is given by
$\xi\rightarrow\chi_{p}(\xi x)$, see e.g. \cite[Section 2.3]{A-K-S}.

Given $x=(x_{1},\dots,x_{d}),$ $\xi=(\xi_{1},\dots,\xi_{d})\in\mathbb{Q}
_{p}^{d}$, we set $x\cdot\xi:=\sum_{j=1}^{d}x_{j}\xi_{j}$. If $f\in L^{1}$,
its Fourier transform is defined by
\[
(\mathcal{F}f)(\xi)=\int_{\mathbb{Q}_{p}^{d}}\chi_{p}(\xi\cdot x)f(x)dx,\qquad
\text{for }\xi\in\mathbb{Q}_{p}^{d}.
\]
We will also use the notation $\mathcal{F}_{x\rightarrow\xi}f$ and
$\widehat{f}$ for the Fourier transform of $f$. The Fourier transform can be
extended as a unitary operator onto $L^{2}$, satisfying
\[
(\mathcal{F}(\mathcal{F}f))(\xi)=f(-\xi) \label{FF(f)}%
\]
for every $f\in L^{2}$, see e.g. \cite[Sections 2.3 and 4.8]{A-K-S} and
\cite[Chapter III, Section 2]{Taibleson}.

The Fourier transform $\mathcal{F}\left[  W\right]  $ of a distribution
$W\in\mathcal{D}^{\prime}\left(  \mathbb{Q}_{p}^{d}\right)  $ is defined by%
\[
\left(  \mathcal{F}\left[  W\right]  ,\varphi\right)  =\left(  W,\mathcal{F}%
\left[  \varphi\right]  \right)  \qquad\text{for all }\varphi\in
\mathcal{D}(\mathbb{Q} _{p}^{N}).
\]
The Fourier transform $W\rightarrow\mathcal{F}\left[  W\right]  $ is a linear
isomorphism from $\mathcal{D}^{\prime}\left(
\mathbb{Q}
_{p}^{N}\right)  $\ onto itself. Furthermore, $W\left(  \xi\right)
=\mathcal{F}\left[  \mathcal{F}\left[  W\right]  \left(  -\xi\right)  \right]
$. We also use the notation $\mathcal{F}_{x\rightarrow\xi}W$ and $\widehat{W}$
for the Fourier transform of $W.$

\section{\label{Section_Taibleson_operator}The Taibleson operator}

We set $\Gamma_{p}^{\left(  d\right)  }(\alpha):=\frac{1-p^{\alpha-d}%
}{1-p^{-\alpha}}$, $\alpha\neq0$. The function
\[
k_{\alpha}(x)=\frac{||x||_{p}^{\alpha-d}}{\Gamma_{p}^{\left(  d\right)
}\left(  \alpha\right)  },\quad\alpha\in\mathbb{R\setminus}\left\{
0,d\right\}  ,\quad x\in\mathbb{Q}_{p}^{d},
\]
is called the \textit{multi-dimensional Riesz kernel}; it determines a
distribution on $\mathcal{D}(\mathbb{Q}_{p}^{d})$ as follows. If
$\alpha\not \in \{0, d\}$, and $\varphi\in\mathcal{D}(\mathbb{Q}_{p}^{d})$,
then
\begin{equation}%
\begin{split}
\left(  k_{\alpha},\varphi\right)   &  =\frac{1-p^{-d}}{1-p^{\alpha-d}}%
\varphi(0)+\frac{1-p^{-\alpha}}{1-p^{\alpha-d}}\int_{||x||_{p}>1}%
||x||_{p}^{\alpha-d}\varphi(x)\,dx\\
&  \quad+\frac{1-p^{-\alpha}}{1-p^{\alpha-d}}\int_{||x||_{p}\leq1}%
||x||_{p}^{\alpha-d}(\varphi(x)-\varphi(0))\,dx.
\end{split}
\label{5}%
\end{equation}
Then $k_{\alpha}\in\mathcal{D}^{\prime}(\mathbb{Q}_{p}^{d})$, for
$\mathbb{R\setminus}\left\{  0,d\right\}  $. In the case $\alpha=0$, by
passing to the limit in (\ref{5}), we obtain
\[
\left(  k_{0},\varphi\right)  :=\lim_{\alpha\rightarrow0}\left(  k_{\alpha
},\varphi\right)  =\varphi(0),
\]
i.e., $k_{0}(x)=\delta\left(  x\right)  $, the Dirac delta distribution, and
therefore $k_{\alpha}\in\mathcal{D}^{\prime}(\mathbb{Q}_{p}^{d})$, for
$\mathbb{R\setminus}\left\{  d\right\}  $.

It follows from (\ref{5}) that for $\alpha>0$,
\begin{equation}
\left(  k_{-\alpha},\varphi\right)  =\frac{1-p^{\alpha}}{1-p^{-\alpha-d}}%
\int_{\mathbb{Q}_{p}^{n}}||x||_{p}^{-\alpha-d}(\varphi(x)-\varphi(0))\,dx.
\label{7}%
\end{equation}

\begin{definition}
The Taibleson pseudodifferential operator $\boldsymbol{D}^{\alpha}$,
$\alpha>0$, is defined as
\[
\boldsymbol{D}^{\alpha}\varphi(x)=\mathcal{F}_{\xi\rightarrow x}^{-1}\left(
||\xi||_{p}^{\alpha}\mathcal{F}_{x\rightarrow\xi}\varphi\right)
,\qquad\text{for }\varphi\in\mathcal{D}. \label{Taibleson_operator}%
\]

\end{definition}

This operator was introduced in \cite{Taibleson}, see also
\cite{Rodriguez-Zuniga} and \cite[Chapter 9]{A-K-S}. The Taibleson operator
coincides with the Vladimirov operator in dimension one.

From the fact that $\left(  \mathcal{F}k_{\alpha}\right)  \left(  x\right)  $,
with $\alpha\neq d$, equals to $||x||_{p}^{-\alpha}$ in $\mathcal{D}^{\prime}%
$, see e.g. \cite[Chap. III, Theorem 4.5]{Taibleson}, and (\ref{7}), we have
\begin{equation}
\boldsymbol{D}^{\alpha}\varphi\left(  x\right)  =\left(  k_{-\alpha}%
\ast\varphi\right)  \left(  x\right)  =\frac{1-p^{\alpha}}{1-p^{-\alpha-d}%
}\int_{\mathbb{Q}_{p}^{d}}||y||_{p}^{-\alpha-d}(\varphi(x-y)-\varphi(x))\,dy.
\label{8}%
\end{equation}
The right-hand side of (\ref{8}) makes sense for a wider class of functions,
for example, for locally constant functions $\varphi$ satisfying
\[
\int_{||x||_{p}\geq1}||x||_{p}^{-\alpha-d}|\varphi(x)|\,dx<\infty.
\]
Consequently, we may assume that the constant functions are contained in the
domain of $\boldsymbol{D}^{\alpha}$, and that $\boldsymbol{D}^{\alpha}%
\varphi=0$, for any constant function $\varphi$. Later on, we will work with
the following extension of $\boldsymbol{D}^{\alpha}$:%
\[%
\begin{array}
[c]{ccc}%
\operatorname{Dom}(\boldsymbol{D}^{\alpha}) & \rightarrow & \mathcal{D}%
^{\prime}\smallskip\\
T & \rightarrow & \boldsymbol{D}^{\alpha}T,
\end{array}
\]
where $\operatorname{Dom}(\boldsymbol{D}^{\alpha}):=\left\{  T\in
\mathcal{D}^{\prime};||x||_{p}^{\alpha}\mathcal{F}\left(  T\right)
\in\mathcal{D}^{\prime}\right\}  $, and $\boldsymbol{D}^{\alpha}%
T=\mathcal{F}^{-1}(||x||_{p}^{\alpha}\mathcal{F}\left(  T\right)  )$. Notice
that for this operator, the formula $\boldsymbol{D}^{\alpha}T=k_{-\alpha}\ast
T$ holds.

\subsection{$p$-adic heat equations}

In this article the Taibleson operator $\boldsymbol{D}^{\alpha}$ will be
considered as a $p$-adic analog of the Laplacian $\Delta:=\frac{\partial^{2}%
}{\partial x_{1}^{2}}+\cdots+\frac{\partial^{2}}{\partial x_{d}^{2}}$\ in
$\mathbb{R}^{d}$. To explain this analogy, we use the `$p$-adic heat
equation,' which is defined as%
\begin{equation}
\frac{\partial u(x,t)}{\partial t}+\boldsymbol{D}^{\alpha}u(x,t)=0,\quad
x\in\mathbb{Q}_{p}^{d},\quad t>0. \label{Eq_1}%
\end{equation}
The analogy with the classical heat equation comes from the fact that the
solution of the initial value problem attached to (\ref{Eq_1}) with initial
datum $u(x,0)=\varphi(x)\in\mathcal{D}_{\mathbb{R}}$ is given by%
\[
u(x,t)=\int_{\mathbb{Q}_{p}^{d}}Z(x-y,t)\varphi(x)\,dx,
\]
where%
\[
Z(x,t):=\int_{\mathbb{Q}_{p}^{d}}\chi_{p}(-x\cdot\xi)e^{-t||\xi||_{p}^{\alpha
}}\,d\xi\qquad\text{for }t>0,
\]
is the $p$\textit{-adic heat kernel}. $Z\left(  x,t\right)  $ is a transition
density of a time and space homogeneous Markov process which is bounded, right
continuous and has no discontinuities other than jumps, cf. \cite[Theorem
16]{Zuniga-LNM-2016}.

The family of `$p$-adic Laplacians' is very large, see e.g. \cite[Chapter
9]{A-K-S}, \cite[Chapter 12]{KKZuniga}, \cite[Chapter 4]{Koch},
\cite{Torresblanca-Zuniga}, \cite[Chapter 2]{Zuniga-LNM-2016} and the
references therein. We pick the Taibleson operator due to the fact that the
corresponding fundamental solutions are well-known.

\subsection{Fundamental solutions}

The $p$-adic analog of the electrostatic equation is%
\begin{equation}
\boldsymbol{D}^{\alpha}u(x)=\varphi\left(  x\right)  ,\qquad\varphi
\in\mathcal{D}. \label{Eq_2}%
\end{equation}

A \textit{fundamental solution} of (\ref{Eq_2}) is a distribution $G_{\alpha}$
such that $u=G_{\alpha}\ast\varphi$ is a solution of (\ref{Eq_2}) in
$\mathcal{D}^{\prime}$.

\begin{proposition}
[{\cite[Theorem 13]{Rodriguez-Zuniga-2010}}]A fundamental solution for
(\ref{Eq_2}) is given by
\[
G_{\alpha}(x)=%
\begin{cases}
\dfrac{1-p^{-\alpha}}{1-p^{\alpha-d}}||x||_{p}^{\alpha-d}, & \text{if }
\alpha\neq d\\
\dfrac{1-p^{d}}{p^{d}\ln p}\ln||x||_{p}, & \text{if }\alpha=d.
\end{cases}
\]

\end{proposition}

By using that $\boldsymbol{D}^{\alpha}{\LARGE \cdot}=k_{-\alpha}%
\ast{\LARGE \cdot}$, and that $\boldsymbol{D}^{\alpha}\left(  G_{\alpha}%
\ast\varphi\right)  =\varphi$ in $\mathcal{D}^{\prime}$, for any $\varphi
\in\mathcal{D}$, we have%
\[
\boldsymbol{D}^{\alpha}\left(  G_{\alpha}\ast\varphi\right)  =k_{-\alpha}%
\ast\left(  G_{\alpha}\ast\varphi\right)  =\left(  k_{-\alpha}\ast G_{\alpha
}\right)  \ast\varphi=\varphi,\qquad\text{in }\mathcal{D}^{\prime},
\]
for any test function $\varphi$, and consequently $k_{-\alpha}\ast G_{\alpha
}=\delta$, i.e.
\[
\boldsymbol{D}^{\alpha}G_{\alpha}=\delta\qquad\text{in }\mathcal{D}^{\prime}.
\label{Eq_3A}%
\]

To allow an easy comparison with the literature on Coulomb gases, we set:
\begin{equation}
g_{\alpha}\left(  x\right)  =%
\begin{cases}
||x||_{p}^{\alpha-d}, & \text{if }\alpha\neq d\\
\ln||x||_{p}, & \text{if }\alpha=d,
\end{cases}
\label{Eq_3}%
\end{equation}
then
\begin{equation}
\boldsymbol{D}^{\alpha}g_{\alpha}=-C_{d,\alpha}\delta\qquad\text{with
}C_{d,\alpha}=
\begin{cases}
\frac{p^{\alpha-d}-1}{1-p^{-\alpha}}, & \text{if }\alpha\neq d\\
\frac{p^{d}\ln p}{p^{d}-1}, & \text{if }\alpha= d.
\end{cases}
\label{Eq_3B}%
\end{equation}
Notice that in the Archimedean case $\alpha=2$, while in the non-Archimedean
case, we have a family of Green functions depending on the parameter $\alpha$.
In addition, in the $p$-adic case, the potentials $||x||_{p}^{\alpha-d}$,
$\ln||x||_{p}$ occur in all the dimensions.

From now on, we assume that $g_{\alpha}\left(  x\right)  =\frac{1}%
{||x||_{p}^{d-\alpha}}$ with $d>\alpha>0$.

\section{Some technical results}

\begin{lemma}
\label{Lemma_1}For $x$, $y\in\mathbb{Q}_{p}^{d}$, with $x\neq y$, and
$\alpha>0$, with $d>\alpha$, we set%
\begin{equation}
\mathcal{I}(x,y,\alpha):=\int_{\mathbb{Q}_{p}}\,\int_{\mathbb{Q}_{p}%
^{d}}\left\vert t\right\vert _{p}^{2d-\alpha-1}\Omega\left(  \left\Vert
t\left(  z-x\right)  \right\Vert _{p}\right)  \Omega\left(  \left\Vert
t\left(  z-y\right)  \right\Vert _{p}\right)  dz\,dt. \label{Eq_8}%
\end{equation}
Then
\[
\mathcal{I}(x,y,\alpha)=\left(  \frac{1-p^{-1}}{1-p^{\alpha-d}}\right)
\left\Vert x-y\right\Vert _{p}^{\alpha-d}.
\]

\end{lemma}

\begin{proof}
The announced formula is proved by a sequence of changes of variables. By
changing variables as $z\rightarrow t^{-1}w+y$, $t\rightarrow t$ (then
$dz\,dt\rightarrow\left\vert t\right\vert _{p}^{-d}dw\,dt$) in (\ref{Eq_8}),
and using that $\Omega\ast\Omega=\Omega$, we get
\begin{align*}
\mathcal{I}(x,y,\alpha)  &  =\int_{\mathbb{Q}_{p}\smallsetminus\left\{
0\right\}  }\,\int_{\mathbb{Q}_{p}^{d}}\left\vert t\right\vert _{p}%
^{d-\alpha-1}\Omega\left(  \left\Vert w\right\Vert _{p}\right)  \Omega\left(
\left\Vert t\left(  x-y\right)  -w\right\Vert _{p}\right)  dw\,dt\\
&  =\int_{\mathbb{Q}_{p}\smallsetminus\left\{  0\right\}  }\,\left\vert
t\right\vert _{p}^{d-\alpha-1}\left(  \Omega\ast\Omega\right)  \left(
\left\Vert t\left(  x-y\right)  \right\Vert _{p}\right)  dt\\
&  =\int_{\mathbb{Q}_{p}\smallsetminus\left\{  0\right\}  }\,\left\vert
t\right\vert _{p}^{d-\alpha-1}\Omega\left(  \left\vert t\left\Vert \left(
x-y\right)  \right\Vert _{p}^{-1}\right\vert _{p}\right)  dt.
\end{align*}
Finally, we change the variables as $t\left\Vert \left(  x-y\right)
\right\Vert _{p}^{-1}\rightarrow s$ (then $dt\rightarrow\left\Vert \left(
x-y\right)  \right\Vert _{p}^{-1}ds$) to obtain
\[
\mathcal{I}(x,y,\alpha)=\left\Vert \left(  x-y\right)  \right\Vert
_{p}^{\alpha-d}\int_{\mathbb{Z}_{p}\smallsetminus\left\{  0\right\}  }\,\left\vert s\right\vert _{p}^{d-\alpha-1}ds=\frac{1-p^{-1}}{1-p^{\alpha-d}
}\left\Vert \left(  x-y\right)  \right\Vert _{p}^{\alpha-d}\quad\text{for
}d>\alpha.\qedhere
\]

\end{proof}

Let $\mu$, $\nu$ be signed Radon measures on $\mathbb{Q}_{p}^{d}$. We set, for
$d>\alpha$,
\[
\mathcal{E}_{\alpha}(\mu,\nu):=\int_{\mathbb{Q}_{p}^{d}}\, \int_{\mathbb{Q}%
_{p}^{d}}\left\Vert x-y\right\Vert _{p}^{\alpha-d}d\mu\left(  x\right)
d\nu\left(  y\right)  .
\]

\begin{proposition}
\label{Prop_1} If $\mathcal{E}_{\alpha}(\left\vert \mu\right\vert ,\left\vert
\mu\right\vert )<+\infty$, then
\begin{equation}
\mathcal{E}_{\alpha}(\mu,\mu)\geq0. \label{Eq_4}%
\end{equation}
The equality in (\ref{Eq_4}) holds if and only if $\mu=0$. Moreover, if
$\mathcal{E}_{\alpha}(\left\vert \nu\right\vert ,\left\vert \nu\right\vert
)<+\infty$, we have the inequality%
\begin{equation}
\left\{  \mathcal{E}_{\alpha}(\mu,\nu)\right\}  ^{2}\leq\mathcal{E}_{\alpha
}(\mu,\mu)\mathcal{E}_{\alpha}(\nu,\nu), \label{Eq_5}%
\end{equation}
with the equality for $\nu\neq0$ if and only if $\mu=c\nu$ for some constant
$c$. The map $\mu\rightarrow\mathcal{E}_{\alpha}(\mu,\mu)$ is strictly convex,
i.e. when $\mu\neq\nu$ and $0<\lambda<1$,%
\begin{equation}
\mathcal{E}_{\alpha}(\lambda\mu+\left(  1-\lambda\right)  \nu,\lambda
\mu+\left(  1-\lambda\right)  \nu)<\lambda\mathcal{E}_{\alpha}(\mu
,\mu)+\left(  1-\lambda\right)  \mathcal{E}_{\alpha}(\nu,\nu). \label{Eq_6}%
\end{equation}

\end{proposition}

\begin{proof}
By applying Lemma \ref{Lemma_1} and Fubini's theorem, here we use the
hypothesis $\mathcal{E}_{\alpha}(\left\vert \mu\right\vert ,\left\vert
\mu\right\vert )<\infty$, we have%
\begin{equation}
\label{Eq_9}%
\begin{split}
\mathcal{E}_{\alpha}(\mu,\mu)  &  =\frac{1-p^{\alpha-d}}{1-p^{-1}}%
\int_{\mathbb{Q}_{p}^{d}}\,\int_{\mathbb{Q}_{p}^{d}}\mathcal{I}(x,y,\alpha
)d\mu\left(  x\right)  d\mu\left(  y\right) \\
&  =\int_{\mathbb{Q}_{p}\smallsetminus\left\{  0\right\}  }\,\left\vert
t\right\vert _{p}^{2d-\alpha-1}\left\{  \sqrt{\frac{1-p^{\alpha-d}}{1-p^{-1}}%
}\int_{\mathbb{Q}_{p}^{d}}g_{\left\vert t\right\vert _{p}}\left(  z\right)
dz\right\}  ^{2}dt\geq0,
\end{split}
\end{equation}
where%
\[
g_{\left\vert t\right\vert _{p}}\left(  z\right)  :=\int_{\mathbb{Q}_{p}^{d}%
}\,\Omega\left(  \left\Vert t\left(  z-x\right)  \right\Vert _{p}\right)
d\mu\left(  x\right)  =\mu\left(  z\right)  \ast1_{B_{\operatorname{ord}%
(t)}^{d}}\left(  z\right)  ,\qquad\text{for }t\neq0.
\]
Now, assume that $\mathcal{E}_{\alpha}(\mu,\mu)=0$. Then, it follows from
(\ref{Eq_9}) that $\mu\left(  z\right)  \ast1_{B_{\operatorname{ord}(t)}^{d}%
}\left(  z\right)  \equiv0$ (which is a locally constant function) for almost
all $t\in\mathbb{Q}_{p}$. This last function is locally constant in $z$ (for
almost all $t)$, and its Fourier transform $\widehat{\mu\ast
1_{B_{\operatorname{ord}(t)}^{d}}}\left(  \xi\right)  =p^{d \operatorname{ord}%
(t)}\widehat{\mu}\left(  \xi\right)  \cdot\Omega\left(  \left\Vert
p^{-\operatorname{ord}(t)}\xi\right\Vert \right)  $, consequently,
$\widehat{\mu}\left(  \xi\right)  =0$ for any $\left\Vert \xi\right\Vert
_{p}\leq\left\vert t\right\vert _{p}$. And since $t\in\mathbb{Q}%
_{p}\smallsetminus\left\{  0\right\}  $ is arbitrary, $\widehat{\mu}=0$, and
thus $\mu=0$.

Inequality (\ref{Eq_5}) is proved by considering $\mathcal{E}_{\alpha}%
(\mu-\lambda\nu,\mu-\lambda\nu)\geq0$, with $\lambda=\frac{\mathcal{E}%
_{\alpha}(\mu,\nu)}{\mathcal{E}_{\alpha}(\nu,\nu)}$.

To prove inequality (\ref{Eq_6}), we use that right hand side minus the left
hand side is equal to
\[
\lambda\left(  1-\lambda\right)  \mathcal{E}_{\alpha}(\mu-\nu,\mu
-\nu).\qedhere
\]

\end{proof}

This proposition is the $p$-adic counterpart of Theorem 9.8 in \cite{Lieb et
al}.

\section{The $p$-adic Coulomb gas}

The Hamiltonian of the $p$-adic Coulomb gas is defined as%
\[
H_{n,\alpha}\left(  x_{1},\ldots,x_{n}\right)  :=H_{n}\left(  x_{1}%
,\ldots,x_{n}\right)  =\sum\limits_{i\neq j}g_{\alpha}\left(  x_{i}%
-x_{j}\right)  +n\sum\limits_{i=1}^{n}V(x_{i}),
\]
where $x_{1},\ldots,x_{n}\in\mathbb{Q}_{p}^{d}$ and $V:\mathbb{Q}_{p}%
^{d}\rightarrow\mathbb{R}$. In this article we only consider the case
$g_{\alpha}\left(  x\right)  =\frac{1}{||x||_{p}^{d-\alpha}}$, with $d>\alpha$.

\subsection{$\Gamma$-convergence}

\begin{definition}
We say that a sequence $\left\{  F_{n}\right\}  _{n\in\mathbb{N}}$ of
functions, on a metric space $X$, $\Gamma$-converges to a function
$F:X\rightarrow\left(  -\infty,+\infty\right]  $ if the following two
inequalities hold:

\begin{itemize}
\item[1.] ($\Gamma-\lim\inf$) if $x_{n}\rightarrow x$ in $X$, then
$\liminf_{n\rightarrow+\infty}F_{n}\left(  x_{n}\right)  \geq F(x)$;

\item[2.] ($\Gamma-\limsup$) for all $x$ in $X$, there exists a sequence
$\left\{  x_{n}\right\}  _{n}$\ in $X$ such that $x_{n}\rightarrow x$ and
$\limsup_{n\rightarrow+\infty}F_{n}\left(  x_{n}\right)  \leq F(x)$. Such a
sequence is called a recovery sequence.
\end{itemize}
\end{definition}

\begin{lemma}
[{\cite[Proposition 2.6]{Serfaty-Book}}]\label{Lemma_2}Assume that $F_{n}$
$\Gamma$-converges to $F$. If for every $n $, $x_{n}$ minimizes $F_{n}$, and
if the sequence $\left\{  x_{n}\right\}  _{n\in\mathbb{N}}$ converges to some
$x$ in $X$, then $x$ minimizes $F$, and moreover, $\lim_{n\rightarrow+\infty
}\min_{X}F_{n}=\min_{X}F$.
\end{lemma}

\subsection{$\Gamma$-convergence of the $p$-adic Coulomb gas Hamiltonian}

We denote by $\mathcal{P}(\mathbb{Q}_{p}^{d})$ the space of probability
measures on $\mathbb{Q}_{p}^{d}$. By using the following map:
\[%
\begin{array}
[c]{ccc}%
\left(  \mathbb{Q}_{p}^{d}\right)  ^{n} & \rightarrow & \mathcal{P}%
(\mathbb{Q}_{p}^{d}) \smallskip\\
\left(  x_{1},\ldots,x_{n}\right)  & \rightarrow & \frac{1}{n}\sum
\limits_{i=1}^{n}\delta_{x_{i}}%
\end{array}
\]
which associates to the configuration of $n$ points the probability measure
$\frac{1}{n}\sum\limits_{i=1}^{n}\delta_{x_{i}}$ (called the \textit{empirical
measure}), here $\delta_{x_{i}}$ denotes the Dirac distribution at $x_{i}$, we
consider $H_{n}\left(  x_{1},\ldots,x_{n}\right)  $ as a function on
$\mathcal{P}(\mathbb{Q}_{p}^{d})$, as follows:
\[
H_{n}\left(  \mu\right)  =
\begin{cases}
H_{n}\left(  x_{1},\ldots,x_{n}\right)  , & \text{if }\mu=\frac{1}{n}
\sum\limits_{i=1}^{n}\delta_{x_{i}}\\
+\infty, & \text{otherwise.}%
\end{cases}
\]

\begin{theorem}
\label{Theo_1A}Assume that $d>\alpha$ and that $V$ is a continuous bounded
from below function. The sequence $\left\{  \frac{1}{n^{2}}H_{n}\right\}
_{n}$ of functions (defined on $\mathcal{P}(\mathbb{Q}_{p}^{d})$) $\Gamma
$-converges, with respect to the weak convergence of probability measures, to
the function $I_{\alpha}:\mathcal{P}(\mathbb{Q}_{p}^{d})\rightarrow\left(
-\infty,+\infty\right]  $ defined by%
\[
I(\mu):=I_{\alpha}(\mu)=\int_{\mathbb{Q}_{p}^{d}}\,\int_{\mathbb{Q}_{p}^{d}%
}g_{\alpha}\left(  x-y\right)  d\mu\left(  x\right)  d\mu\left(  y\right)
+\int_{\mathbb{Q}_{p}^{d}}V\left(  x\right)  d\mu\left(  x\right)  .
\]

\end{theorem}

\begin{remark}
From the point of view of statistical mechanics, $I$ is the mean-field limit
energy of $H_{n}$.
\end{remark}

The proof of this result will be given in Section \ref{Section_Proof_Prop_7}.

\subsection{Minimizing the mean-field energy via potential theory}

In this section we consider the following minimization problem:%
\begin{equation}
\min_{\mu\in\mathcal{P}(\mathbb{Q}_{p}^{d})}I(\mu)=\min_{\mu\in\mathcal{P}%
(\mathbb{Q}_{p}^{d})}\left\{  \int_{\mathbb{Q}_{p}^{d}}\,\int_{\mathbb{Q}%
_{p}^{d}}g_{\alpha}\left(  x-y\right)  d\mu\left(  x\right)  d\mu\left(
y\right)  +\int_{\mathbb{Q}_{p}^{d}}V\left(  x\right)  d\mu\left(  x\right)
\right\}  . \label{Eq_11}%
\end{equation}

\begin{lemma}
\label{Lemma_3}The functional $I$ is strictly convex on $\mathcal{P}%
(\mathbb{Q}_{p}^{d})$.
\end{lemma}

\begin{proof}
Since $\mu\rightarrow\int Vd\mu$ is linear, it is sufficient to show that the
functional (the \textit{mutual energy} of the measures $\mu$, $\nu$)
\begin{equation}
\mathcal{E}_{\alpha}(\mu,\nu):=\int_{\mathbb{Q}_{p}^{d}}\, \int_{\mathbb{Q}%
_{p}^{d}}g_{\alpha}\left(  x-y\right)  d\mu\left(  x\right)  d\mu\left(
y\right)  \label{Eq_Energy}%
\end{equation}
satisfies%
\[
\mathcal{E}_{\alpha}(\lambda\mu+\left(  1-\lambda\right)  \nu,\lambda
\mu+\left(  1-\lambda\right)  \nu)<\lambda\mathcal{E}_{\alpha}(\mu
,\mu)+\left(  1-\lambda\right)  \mathcal{E}_{\alpha}(\nu,\nu), \label{Eq_12}%
\]
for $0<\lambda<1$ and $\mu$, $\nu$ belonging to the convex cone of probability
measures, with $\mathcal{E}_{\alpha}(\mu,\mu)$, $\mathcal{E}_{\alpha}(\nu
,\nu)<+\infty$. This fact follows from Proposition \ref{Prop_1}.
\end{proof}

As a consequence, if there exists a minimizer to (\ref{Eq_11}), it is unique.
This minimizer is called the \textit{equilibrium measure} or the
\textit{Frostman equilibrium measure} in potential theory. In order to show
the existence of an equilibrium measure we make the following assumptions on
the potential $V$:
\begin{gather}
V\text{ is lower semi-continuous and bounded from below function;}\tag{A1}\\
\lim_{\left\Vert x\right\Vert _{p}\rightarrow+\infty}\left(  V(x)+g_{\alpha
}(x)\right)  =+\infty. \tag{A2}%
\end{gather}

The first condition assures the lower semi-continuity of $I$ and that $\inf
I>-\infty$. The second condition is equivalent to the condition $\lim
_{\left\Vert x\right\Vert _{p}\rightarrow+\infty}V(x)=+\infty$.

\begin{lemma}
[{\cite[Lemma 2.10]{Serfaty-Book}}]\label{Lemma_4}Assume that (A1) and (A2)
are satisfied. Let $\left\{  \mu_{n}\right\}  _{n}$ be a sequence in
$\mathcal{P}(\mathbb{Q}_{p}^{d})$ such that $\left\{  I\left(  \mu_{n}\right)
\right\}  _{n}$ is bounded. Then, up to\ extraction of a subsequence,
$\left\{  \mu_{n}\right\}  _{n}$ converges to some $\mu$ in $\mathcal{P}%
(\mathbb{Q}_{p}^{d})$ in the weak sense of probabilities and
\[
\liminf_{n\rightarrow\infty} I\left(  \mu_{n}\right)  \geq I(\mu).
\]

\end{lemma}

\begin{definition}
We define the capacity of a compact set $K\subset\mathbb{Q}_{p}^{d}$ by
\[
\operatorname{Cap}_{\alpha}(K)=\frac{1}{\inf_{\mu\in\mathcal{P}(K)}%
\mathcal{E}_{\alpha}(\mu,\mu)},\qquad\text{with }d>\alpha>0,
\]
where $\mathcal{P}(K)$ denotes the set of probability measures supported in
$K$, and $\mathcal{E}_{\alpha}(\mu,\mu)$ denotes the Coulomb energy defined as
in (\ref{Eq_Energy}). The capacity of $K$ is $+\infty$ if there no exists a
probability measure $\mu\in\mathcal{P}(K)$ such that $\mathcal{E}_{\alpha}%
(\mu,\mu)<+\infty$. For a general set $E\subset\mathbb{Q}_{p}^{d}$, we set%
\[
\operatorname{Cap}_{\alpha}(E)=\sup_{K\subset E}\operatorname{Cap}_{\alpha
}(K),
\]
where $K$ runs through all the compact subsets of $E$.
\end{definition}

Alternatively, we can define the capacity of an arbitrary set $A\subset
\mathbb{Q}_{p}^{d}$ as $\operatorname{Cap}_{\alpha}(A)=\frac{1}{\inf
\mathcal{E}_{\alpha}(\mu,\mu)}$, where $\mu$ runs through all the positive
measures concentrated on $A$ with total mass $\mu\left(  \mathbb{Q}_{p}%
^{d}\right)  =\mu\left(  A\right)  =1$. The result would be the same if the
support of $\mu$ is required to be compact and contained in $A$, see e.g.
\cite[Lemma 2.2.2]{Fuglede}.

The capacity is an increasing function. In addition, it satisfies the following:

\begin{lemma}
[{\cite[Lemma 2.3.1]{Fuglede}}]Let $N$ be a subset of $\mathbb{Q}_{p}^{d}$.
The following conditions are equivalent:

\begin{itemize}
\item[(i)] $\operatorname{Cap}_{\alpha}(N)=0$;

\item[(ii)] $\mu=0$ is the only positive measure of finite energy (i.e.
$\mathcal{E}_{\alpha}(\mu,\mu)<+\infty$) concentrated on $N$;

\item[(iii)] $\mu=0$ is the only positive measure of finite energy supported
by some compact subset of $N$.
\end{itemize}
\end{lemma}

A property is said to hold \textit{quasi-everywhere} (q.e.), if it holds
everywhere except on a set of capacity zero.

In order to allow the potential $V$ to take the value $+\infty$, which is
equivalent to impose the constraint that the probability measures are
supported only on a specific set, the set where $V$ is finite, we impose on
the potential $V$ the following condition:
\begin{equation}
\left\{  x\in\mathbb{Q}_{p}^{d};V(x)<+\infty\right\}  \quad\text{has positive
capacity.} \tag{A3}%
\end{equation}

\begin{lemma}
[{\cite[Lemma 2.13]{Serfaty-Book}}]\label{Lemma_5}Under the assumptions
(A1)--(A3), we have
\[
\inf I<+\infty.
\]

\end{lemma}

\begin{theorem}
\label{Theorem_1}Under the assumptions (A1)--(A3) and $d>\alpha$, the minimum
of $I$ over $\mathcal{P}(\mathbb{Q}_{p}^{d})$ exists, is finite and is
achieved by a unique $\mu_{0}$, which has compact support of positive
capacity. In addition $\mu_{0}$ is uniquely characterized by the fact that%
\begin{equation}%
\begin{cases}
h_{\alpha,\mu_{0}}+\frac{V}{2}\geq C & \text{q.e.\ in }\mathbb{Q}_{p}^{d}\\
h_{\alpha,\mu_{0}}+\frac{V}{2}=C & \text{q.e.\ in the support of }\mu_{0},
\end{cases}
\label{Eq_15}%
\end{equation}
where
\begin{equation}
h_{\alpha,\mu_{0}}(x):=\int_{\mathbb{Q}_{p}^{d}}g_{\alpha}\left(  x-y\right)
d\mu_{0}\left(  y\right)  \label{Eq_16}%
\end{equation}
is the electrostatic potential generated by $\mu_{0}$, and
\begin{equation}
C:=I(\mu_{0})-\frac{1}{2}\int_{\mathbb{Q}_{p}^{d}}V(x)d\mu_{0}\left(
x\right)  . \label{Eq_17}%
\end{equation}

\end{theorem}

\begin{proof}
The proof of this theorem is a slight variation of the proof of Theorem 2.1 in
Serfaty's book \cite{Serfaty-Book}. The reason is that the argument given in
\cite{Serfaty-Book} \ works on a Polish space. More precisely, we need
$\left(  \mathbb{Q}_{p}^{d},\left\Vert \cdot\right\Vert _{p}\right)  $ to be a
complete countable metric space, in order to use Prokhorov's theorem, see e.g.
\cite{Billinsgley}, and that a subset $K$ of $\mathbb{Q}_{p}^{d}$ is compact
if and only if it is closed and bounded. We give just some comments about the
proof. For the details the reader may consult \cite{Serfaty-Book}.

Set $I:=\inf_{\mu\in\mathcal{P}(\mathbb{Q}_{p}^{d})}I(\mu)$. Then by Lemma
\ref{Lemma_5}, $\inf I<+\infty$, and there exists a sequence $\left\{  \mu
_{n}\right\}  _{n}$ in $\mathcal{P}(\mathbb{Q}_{p}^{d})$ such that $I(\mu
_{n})\rightarrow I$. Since the sequence $\left\{  I(\mu_{n})\right\}  _{n}$ is
bounded, by Lemma \ref{Lemma_4}, there exists a probability measure $\mu_{0}$
such that a subsequence of $\{\mu_{n}\}_{n\in\mathbb{N}}$ converges in
probability to $\mu_{0}$ and
\[
I(\mu_{0})\leq\liminf_{n\rightarrow\infty}I(\mu_{n})\leq I.
\]
Then, by the definition of $I$, $I(\mu_{0})=I$. The uniqueness of $\mu_{0}$
follows from Lemma \ref{Lemma_3}. The proof now follows as in \cite[Theorem
2.1]{Serfaty-Book}.
\end{proof}

\section{\label{Section_Coulomb_gas_unit_ball}The Coulomb gas confined into
the unit ball}

We denote by $GL\left(  \mathbb{Z}_{p},d\right)  $ the group of all the
matrices $\boldsymbol{g}$\ of size $d\times d$ with entries in $\mathbb{Z}%
_{p}$ satisfying $\left\vert \det\boldsymbol{g}\right\vert _{p}=1$. This group
preserves the norm $\left\Vert \cdot\right\Vert _{p}$, i.e.
\[
\left\Vert \boldsymbol{g}x\right\Vert _{p}=\left\Vert x\right\Vert _{p}%
\qquad\text{for any }\boldsymbol{g}\mathfrak{\in}GL\left(  \mathbb{Z}%
_{p},d\right)  \text{ and any }x\in\mathbb{Q}_{p}^{d},
\]
see e.g. \cite[Lemma 3.16]{KKZuniga}. Given $\boldsymbol{g}\in GL\left(
\mathbb{Z}_{p},d\right)  $, we define the probability measure
\[
\mu_{\boldsymbol{g}}\left(  B\right)  =\mu_{0}\left(  \boldsymbol{g}%
^{-1}B\right)  ,\qquad\text{with }B\text{ a Borel subset of }\mathbb{Q}%
_{p}^{d},
\]
where $\mu_{0}$\ is the equilibrium measure given in Theorem \ref{Theorem_1}.
Then $\mu_{\boldsymbol{g}}$ is a probability measure supported in
$\boldsymbol{g}^{-1}K$, with $K=\operatorname{supp}\mu_{0}$. Furthermore,
\[
\int\nolimits_{K}f\left(  y\right)  d\mu_{0}\left(  y\right)  =\int
\nolimits_{\boldsymbol{g}^{-1}K}f(\boldsymbol{g}z)d\mu_{\boldsymbol{g}}\left(
z\right)  .
\]
If the potential is a radial function, i.e.\ $V(x)=V(\left\Vert x\right\Vert
_{p})$, then
\[
\int\nolimits_{\boldsymbol{g}^{-1}K}V(\left\Vert x\right\Vert _{p}%
)d\mu_{\boldsymbol{g}}\left(  z\right)  =\int\nolimits_{\boldsymbol{g}^{-1}%
K}V(\left\Vert \boldsymbol{g}x\right\Vert _{p})d\mu_{\boldsymbol{g}}\left(
z\right)  =\int\nolimits_{K}V(\left\Vert x\right\Vert _{p})d\mu_{0}\left(
z\right)  .
\]
Now we set
\[
h_{\alpha,\mu_{\boldsymbol{g}}}(x):=\int\nolimits_{\boldsymbol{g}^{-1}%
K}g_{\alpha}(x-\boldsymbol{g}z)d\mu_{\boldsymbol{g}}\left(  z\right)  .
\]
Then
\[
h_{\alpha,\mu_{\boldsymbol{g}}}(x)=\int\nolimits_{K}g_{\alpha}(x-z)d\mu
_{0}\left(  z\right)  .
\]
Consequently, the measure satisfies conditions (\ref{Eq_15}), and by the
uniqueness of $\mu_{0}$, we conclude that%
\[
\boldsymbol{g}K=K\qquad\text{for any }\boldsymbol{g}\in GL\left(
\mathbb{Z}_{p},d\right)  .
\]
Which implies that
\begin{equation}
K=\bigsqcup\limits_{j\in L}S_{j}^{d},\label{Eq_decomposition}%
\end{equation}
i.e. the support $K$ of the measure $\mu_{0}$ is a union of spheres.

\begin{proposition}
\label{Prop_7}Consider the potential%
\begin{equation}
V(x)=
\begin{cases}
V_{0}, & \text{if } \left\Vert x\right\Vert _{p}\leq1\\
+\infty, & \text{if } \left\Vert x\right\Vert _{p}>1,
\end{cases}
\label{Eq_20}%
\end{equation}
where $V_{0}$ is a positive real number. Then the equilibrium measure is the
characteristic function of the unit ball, i.e. $\mu_{0}\left(  x\right)
=\Omega\left(  \left\Vert x\right\Vert _{p}\right)  $, and
\begin{equation}
I(\mu_{0})=V_{0}+\frac{1-p^{-\alpha}}{1-p^{-d}}. \label{Eq_20AA}%
\end{equation}

\end{proposition}

\begin{proof}
The support of the equilibrium measure is contained in $\mathbb{Z}_{p}^{d}$.
Notice that by the discussion presented at the beginning of this section we
cannot conclude that the support of $\mu_{0}$ is the unit ball, see
(\ref{Eq_decomposition}). So we proceed as follows. We compute a candidate to
the equilibrium measure assuming that $\operatorname{supp}\mu_{0}%
=\mathbb{Z}_{p}^{d}$, then we verify that the proposed measure satisfies
conditions (\ref{Eq_15}).

By restricting the formula given in (\ref{Eq_15}) to the unit ball, we have%
\begin{equation}
\Omega\left(  \left\Vert x\right\Vert _{p}\right)  h_{\alpha,\mu_{0}%
}(x)=\left(  C-\frac{V_{0}}{2}\right)  \Omega\left(  \left\Vert x\right\Vert
_{p}\right)  . \label{Eq_20A}%
\end{equation}
We apply the operator $\boldsymbol{D}^{\alpha}$, with domain $\left\{
T\in\mathcal{D}^{\prime}:\left\Vert \xi\right\Vert _{p}^{\alpha}\widehat{T}%
\in\mathcal{D}^{\prime}\right\}  $, to both sides in (\ref{Eq_20A}). We first
notice that $g_{\alpha}\left(  x\right)  =\frac{1}{\left\Vert x\right\Vert
_{p}^{d-\alpha}}$, $d>\alpha$, is locally integrable, and since $\mu
_{0}\left(  x\right)  $ has compact support, then $\mu_{0}\ast g_{\alpha}%
\in\mathcal{D}^{\prime}$, and $\widehat{\mu_{0}\ast g_{\alpha}}\left(
\xi\right)  =\widehat{\mu_{0}}\left(  \xi\right)  \widehat{g_{\alpha}}\left(
\xi\right)  $. Furthermore, since the support of $\mu_{0}\left(  x\right)  $
by supposition is the unit ball, then
\[
\mathcal{T}_{y}\widehat{\mu_{0}}\left(  \xi\right)  =\widehat{\mu_{0}}\left(
\xi\right)  \qquad\text{for any }y\in\mathbb{Z}_{p}^{d},
\]
where $\mathcal{T}_{y}$ is the translation operator defined as $\mathcal{T}%
_{y}\varphi\left(  x\right)  =\varphi\left(  x-y\right)  $, $\varphi
\in\mathcal{D}$, and $\left(  \mathcal{T}_{y}G,\varphi\right)  =\left(
G,\mathcal{T}_{-y}\varphi\right)  $, $\varphi\in\mathcal{D}$, $G\in
\mathcal{D}^{\prime}$, see e.g. \cite[Chap. III, Proposition 3.17]{Taibleson}.
Then
\begin{align*}
\widehat{\mu_{0}}\left(  \xi\right)  \widehat{g_{\alpha}}\left(  \xi\right)
\ast\Omega\left(  \left\Vert \xi\right\Vert _{p}\right)   &  =\left(
\widehat{\mu_{0}}\left(  y\right)  \widehat{g_{\alpha}}\left(  y\right)
,\Omega\left(  \left\Vert \xi-y\right\Vert _{p}\right)  \right) \\
=\left(  \widehat{\mu_{0}}\left(  y\right)  \widehat{g_{\alpha}}\left(
y\right)  ,\mathcal{T}_{-\xi}\Omega\left(  \left\Vert -y\right\Vert
_{p}\right)  \right)   &  =\left(  \mathcal{T}_{\xi}\widehat{\mu_{0}}\left(
y\right)  \,\mathcal{T}_{\xi}\widehat{g_{\alpha}}\left(  y\right)
,\Omega\left(  \left\Vert -y\right\Vert _{p}\right)  \right) \\
=\left(  \mathcal{T}_{y}\widehat{\mu_{0}}\left(  \xi\right)  \,\mathcal{T}%
_{\xi}\widehat{g_{\alpha}}\left(  y\right)  ,\Omega\left(  \left\Vert
-y\right\Vert _{p}\right)  \right)   &  =\widehat{\mu_{0}}\left(  \xi\right)
\left(  \mathcal{T}_{\xi}\widehat{g_{\alpha}}\left(  y\right)  ,\Omega\left(
\left\Vert -y\right\Vert _{p}\right)  \right) \\
=\widehat{\mu_{0}}\left(  y\right)  \left(  \widehat{g_{\alpha}}\left(
y\right)  ,\mathcal{T}_{-\xi}\Omega\left(  \left\Vert -y\right\Vert
_{p}\right)  \right)   &  =\widehat{\mu_{0}}\left(  \xi\right)  \left(
\widehat{g_{\alpha}}\left(  \xi\right)  \ast\Omega\left(  \left\Vert
\xi\right\Vert _{p}\right)  \right)  .
\end{align*}

We now compute
\begin{equation}
\label{Eq_20B}%
\begin{split}
\mathcal{F}_{\xi\rightarrow x}\left(  \boldsymbol{D}^{\alpha}\Omega\left(
\left\Vert x\right\Vert _{p}\right)  h_{\alpha,\mu_{0}}(x)\right)   &
=\left\Vert \xi\right\Vert _{p}^{\alpha}\,\left(  \widehat{\mu_{0}}\left(
\xi\right)  \widehat{g_{\alpha}}\left(  \xi\right)  \ast\Omega\left(
\left\Vert \xi\right\Vert _{p}\right)  \right) \\
&  =\widehat{\mu_{0}}\left(  \xi\right)  \left\Vert \xi\right\Vert
_{p}^{\alpha}\,\left(  \widehat{g_{\alpha}}\left(  \xi\right)  \ast
\Omega\left(  \left\Vert \xi\right\Vert _{p}\right)  \right)  .
\end{split}
\end{equation}
Then from (\ref{Eq_20A})--(\ref{Eq_20B}), we obtain%
\begin{equation}
\widehat{\mu_{0}}\left(  \xi\right)  \left(  \widehat{g_{\alpha}}\left(
\xi\right)  \ast\Omega\left(  \left\Vert \xi\right\Vert _{p}\right)  \right)
=\left(  C-\frac{V_{0}}{2}\right)  \Omega\left(  \left\Vert \xi\right\Vert
_{p}\right)  , \label{Eq_20C}%
\end{equation}
which implies that the distribution in the left side of (\ref{Eq_20C}) is
supported in the unit ball. And since the product of distributions is
associative,
\[
\widehat{\mu_{0}}\left(  \xi\right)  \left\{  \Omega\left(  \left\Vert
\xi\right\Vert _{p}\right)  \left(  \widehat{g_{\alpha}}\left(  \xi\right)
\ast\Omega\left(  \left\Vert \xi\right\Vert _{p}\right)  \right)  \right\}
=\left(  C-\frac{V_{0}}{2}\right)  \Omega\left(  \left\Vert \xi\right\Vert
_{p}\right)  .
\]
We now use that
\[
\widehat{g_{\alpha}}\left(  \xi\right)  \ast\Omega\left(  \left\Vert
\xi\right\Vert _{p}\right)  =
\begin{cases}
\frac{1-p^{-d}}{1-p^{-\alpha}}, & \text{if } \left\Vert \xi\right\Vert
_{p}\leq1\\
\frac{1}{\left\Vert \xi\right\Vert _{p}^{d-\alpha}}, & \text{if } \left\Vert
x\right\Vert _{p}>1,
\end{cases}
\]
to obtain $\widehat{\mu_{0}}\left(  \xi\right)  =\frac{1-p^{-\alpha}}%
{1-p^{-d}}\left(  C-\frac{V_{0}}{2}\right)  \Omega\left(  \left\Vert
\xi\right\Vert _{p}\right)  $, and hence
\[
\mu_{0}\left(  x\right)  =\frac{1-p^{-\alpha}}{1-p^{-d}}\left(  C-\frac{V_{0}%
}{2}\right)  \Omega\left(  \left\Vert x\right\Vert _{p}\right)  .
\]
Then necessarily $\frac{1-p^{-\alpha}}{1-p^{-d}}\left(  C-\frac{V_{0}}%
{2}\right)  =1$, which implies (\ref{Eq_20AA}). Finally, the verification that
$\mu_{0}\left(  x\right)  =\Omega\left(  \left\Vert x\right\Vert _{p}\right)
$ satisfies (\ref{Eq_15}) is straightforward.
\end{proof}

\section{Proof of the $\Gamma$-convergence and some consequences}

\subsection{\label{Section_Proof_Prop_7}Proof of Theorem \ref{Theo_1A}}

The proof is organized in the same form as the proof of Proposition 2.8 in
\cite{Serfaty-Book}. In the proof the topology of $\mathbb{Q}_{p}^{d}$ comes
into play, and consequently there are important differences with the classical case.

\textbf{Step 1. }($\Gamma-\liminf$) If $\frac{1}{n}\sum_{i=1}^{n}\delta
_{x_{i}}\rightarrow\mu$, then
\[
\liminf_{n\rightarrow+\infty}\frac{1}{n^{2}}H_{n}(x_{1},\ldots,x_{n})\geq
I(\mu).
\]
For the proof of this assertion the reader may consult \cite[pp.
23--24]{Serfaty-Book}.

\textbf{Step 2. }($\Gamma-\limsup$) We have to construct a recovery sequence
for each measure $\mu\in\mathcal{P}(\mathbb{Q}_{p}^{d})$ such that
$I(\mu)<+\infty$. Similarly to the proof of Proposition 2.8 from
\cite{Serfaty-Book} it is sufficient to prove the statement for compactly
supported measures. Moreover, by considering the $\delta$-approximating
sequence
\[
\delta_{n}(x):=
\begin{cases}
p^{nd}, & \text{if } \left\Vert x\right\Vert _{p}\leq p^{-n}\\
0, & \text{if } \left\Vert x\right\Vert _{p}>p^{-n},
\end{cases}
\]
and the convolutions $\mu_{n}=\mu\ast\delta_{n}$ and repeating the
corresponding part of the proof of Proposition 2.8 from \cite{Serfaty-Book},
we may further assume that $\mu$ is supported in some ball $B_{L}^{d}=p^{-L}
\mathbb{Z}_{p}^{d}$, $L\ge0$, has a density in $\mathcal{D}_{\mathbb{R}%
}(\mathbb{Q}_{p}^{d})$, and that this density is bounded from below by
$\epsilon>0$, from above by $p^{Kd}-1$ for some $K\in\mathbb{N}$ and its index
of local constancy $l\left(  \mu\right)  $ satisfies $l\left(  \mu\right)
\geq-M_{0}$ for some $M_{0}\in\mathbb{N}$.

\textbf{Step 3.} Let us fix some $M\ge M_{0}$. There are $p^{(L+M)d}$ balls of
radius $p^{-M}$ in the support of the measure $\mu$. Let us denote them as
$B_{k}$, $1\le k\le p^{(L+M)d}$. In each of these balls consider $p^{(M+K)d}$
smaller balls of radius $p^{-2M-K}$.

Now we distribute $p^{2Md}$ points into these larger balls as follows. In each
ball $B_{k}$ we place $\left[  p^{2Md}\mu(B_{k})\right]  +\epsilon_{k}$
points, here $[x]$ denotes the largest integer not exceeding $x$, and we take
$\epsilon_{k}$ equal to $0$ or $1$ so that the total number of distributed
points equals $p^{2Md}$. The total number of points in the ball $B_{k}$ does
not exceed
\[
p^{2Md}\mu(B_{k})+1 \le p^{2Md}\cdot(p^{Kd}-1) \cdot p^{-Md}+1 \le
p^{(M+K)d},
\]
that is there are sufficiently many smaller balls (of radius $p^{-2M-K}$) to
choose at most one point in each smaller ball. In such way we may select
$p^{2Md}$ points $x_{1},\ldots,x_{p^{2Md}}$ and the distance between any two
points will be at least $p^{-2M-K+1}$.

Consider the measure $\mu_{M}:=p^{-2Md}\sum_{i=1}^{p^{2Md}}\delta_{x_{i}}$ and
let us show that $\mu_{M}\rightharpoonup\mu$, $M\to+\infty$ in the weak sense
of probabilities.

Let us fix a test function $\varphi\in\mathcal{D}_{\mathbb{R}}(\mathbb{Q}%
_{p}^{d})$ and let $M\ge M_{0}$ be such that the index of local constancy
$l\left(  \varphi\right)  $ satisfies $l\left(  \varphi\right)  \geq-M$. Both
the density of the measure $\mu$ and the function $\varphi$ are constant on
the balls $B_{k}$ of radius $p^{-M}$. Denote by $b_{k}$ a point in the ball
$B_{k}$. Then
\[%
\begin{split}
\int_{\mathbb{Q}_{p}^{d}} \varphi(x)\, d\mu_{M}(x) - \int_{\mathbb{Q}_{p}^{d}}
\varphi(x)\, d\mu(x)  &  = \frac1{p^{2Md}}\sum_{i=1}^{p^{2Md}} \varphi(x_{i})
- \sum_{k=1}^{p^{(M+L)d}} \varphi(b_{k})\mu(B_{k})\\
&  =\sum_{k=1}^{p^{(M+L)d}} \varphi(b_{k})\left(  \frac{\left[  p^{2Md}%
\mu(B_{k})\right]  +\epsilon_{k}}{p^{2Md}} - \mu(B_{k})\right)  ,
\end{split}
\]
where $\epsilon_{k}$ denotes 0 or 1 depending on the selection of the points
$x_{i}$. The last expression converges to zero as $M\to+\infty$ since
\[
-\frac{1}{p^{2Md}}\le\frac{\left[  p^{2Md}\mu(B_{k})\right]  +\epsilon_{k}%
}{p^{2Md}} - \mu(B_{k})\le\frac{1}{p^{2Md}}.
\]

\textbf{Step 4.} Let $M$ be fixed and the points $x_{1},\ldots,x_{p^{2Md}}$ be
chosen as in Step 3. We denote by $\Delta$ the diagonal of $\mathbb{Q}_{p}%
^{d}\times\mathbb{Q}_{p}^{d}$, and set $\Delta^{c}:=\mathbb{Q}_{p}^{d}%
\times\mathbb{Q}_{p}^{d}\smallsetminus\Delta$. Then
\[%
\begin{split}
\left(  \frac{1}{p^{2Md}}\right)  ^{2} H_{p^{2Md}}(\mu_{M})  &  = \left(
\frac{1}{p^{2Md}}\right)  ^{2} H_{p^{2Md}}(x_{1},\ldots,x_{p^{2Md}})\\
&  = \frac1{p^{4Md}}\sum_{\substack{i,j=1\\i\ne j}}^{p^{2Md}} g_{\alpha
}\left(  x_{i}-x_{j}\right)  + \frac1{p^{2Md}}\sum_{i=1}^{p^{2Md}}V(x_{i})\\
&  =\iint_{\Delta^{c}}g_{\alpha}(x-y)d\mu_{M}( x) d\mu_{M}( y) +
\int_{\mathbb{Q}_{p}^{d}} V(x)\,d\mu_{M}(x).
\end{split}
\]

Consider
\begin{multline*}
\iint_{\Delta^{c}}g_{\alpha}(x-y)d\mu_{M}( x) d\mu_{M}( y) =\\
\iint_{\Delta^{c}}( \Omega g_{\alpha}) (\left\Vert x-y\right\Vert _{p}%
)d\mu_{M}( x) d\mu_{M}( y) +\iint_{\Delta^{c}}\left[  \left(  1-\Omega\right)
g_{\alpha}\right]  (\left\Vert x-y\right\Vert _{p})d\mu_{M}( x) d\mu_{M}( y)\\
=:E_{M}^{\left(  0\right)  }+E_{M}^{\left(  1\right)  },
\end{multline*}
where $\Omega\bigl(  \left\Vert x\right\Vert _{p}\bigr)  $ is the
characteristic function of $\mathbb{Z}_{p}^{d}$. The function
\[
\Omega\bigl(  p^{-L}\left\Vert x\right\Vert _{p}\bigr)  \left(  1-\Omega
\bigl(  \left\Vert x\right\Vert _{p}\bigr)  \right)  g_{\alpha}\left(
x\right)
\]
is a test function supported in the ball $B_{L}^{d}$, where $\mu_{M}$ and
$\mu$ are supported. By using that $\mu_{M}\rightharpoonup\mu$, we conclude
that $E_{M}^{\left(  1\right)  }$ converges to
\[
\iint_{\Delta^{c}} \left(  \left(  1-\Omega\right)  g_{\alpha}\right)
(\left\Vert x-y\right\Vert _{p})d\mu( x) d\mu( y) .
\]

\textbf{Claim}:
\[
\lim_{M\rightarrow+\infty}E_{M}^{\left(  0\right)  }=\iint_{\Delta^{c}}\left(
\Omega g_{\alpha}\right)  (\left\Vert x-y\right\Vert _{p})d\mu( x) d\mu( y).
\]

Since $\left(  \Omega g_{\alpha}\right)  \bigl(  \left\Vert x\right\Vert
_{p}\bigr)= \frac{\Omega\left(  \left\Vert x\right\Vert _{p}\right)
}{\left\Vert x\right\Vert _{p}^{d-\alpha}}\in L_{\mathbb{R}}^{1}\left(
\mathbb{Z}_{p} ^{d},dx\right)  $ and $\mathcal{D}_{\mathbb{R}}(\mathbb{Z}%
_{p}^{d})$ is dense in $L_{\mathbb{R}}^{1}\left(  \mathbb{Z}_{p}%
^{d},dx\right)  $, given $\epsilon>0$ there exists $\varphi\in\mathcal{D}%
_{\mathbb{R}}(\mathbb{Z}_{p}^{d})$ satisfying
\begin{equation}
\left\Vert \varphi\left(  x\right)  -\frac{\Omega\bigl( \left\Vert
x\right\Vert _{p}\bigr) }{\left\Vert x\right\Vert _{p}^{d-\alpha}}\right\Vert
_{1}<\epsilon. \label{Eq_30}%
\end{equation}

Now, by using (\ref{Eq_30}) and the Young's inequality $\left\Vert \nu\ast
f\right\Vert _{1}\leq\left\Vert f\right\Vert _{1}\left\Vert \nu\right\Vert ,$
where $\nu$ is a finite Borel measure, $\|\nu\|$ is its total variation and
$f\in L_{\mathbb{R}}^{1}$, we obtain
\[%
\begin{split}
\biggl\vert E_{M}^{\left(  0\right)  }  &  -\iint_{\Delta^{c}}\left(  \Omega
g_{\alpha}\right)  (\left\Vert x-y\right\Vert _{p})d\mu( x) d\mu( y)
\biggr\vert\\
&  \leq\left\vert \iint_{\Delta^{c}}\biggl\{ \varphi\left(  x-y\right)
-\frac{\Omega\bigl( \left\Vert x-y\right\Vert _{p}\bigr) }{\left\Vert
x-y\right\Vert _{p}^{d-\alpha}}\biggr\} d\mu_{M}( x) d\mu_{M}( y) \right\vert
\\
&  \quad+\left\vert \iint_{\Delta^{c}}\varphi\left(  x-y\right)  d\mu_{M}( x)
d\mu_{M}( y) -\iint_{\Delta^{c}}\varphi\left(  x-y\right)  d\mu( x) d\mu( y)
\right\vert \\
&  \quad+\left\vert \iint_{\Delta^{c}}\biggl\{ \varphi\left(  x-y\right)
-\frac{\Omega\bigl( \left\Vert x-y\right\Vert _{p}\bigr) }{\left\Vert
x-y\right\Vert _{p}^{d-\alpha}}\biggr\} d\mu( x) d\mu( y) \right\vert ,
\end{split}
\]
hence
\[
\lim_{M\rightarrow+\infty}\left\vert E_{M}^{\left(  0\right)  }-\iint
_{\Delta^{c}}\left(  \Omega g_{\alpha}\right)  (\left\Vert x-y\right\Vert
_{p})d\mu( x) d\mu( y) \right\vert \leq2 \epsilon,
\]
which implies the announced Claim.

Therefore,
\[
\limsup_{M\rightarrow+\infty}\iint_{\Delta^{c}}g_{\alpha}(x-y)d\mu_{M}( x)
d\mu_{M}( y) \leq\iint_{\Delta^{c} }g_{\alpha}(x-y)d\mu( x) d\mu( y)
\]
(actually, equality holds).

On the other hand, since $V$ is continuous, $\mu_{M}\rightharpoonup\mu$, and
$\mu_{M}$, $\mu$ are supported in $B_{L}^{d}$, we have $\int Vd\mu
_{M}\rightarrow\int Vd\mu$.

In conclusion,
\[
\limsup_{M\rightarrow+\infty}\left(  \frac{1}{p^{2Md}}\right)  ^{2}H_{p^{2Md}%
}\left(  \mu_{M}\right)  \leq I(\mu).
\]

\subsection{Some further results}

\begin{lemma}
[{\cite[Lemma 2.21]{Serfaty-Book}}]\label{Lemma_6}Assume that $V$ satisfies
(A1)--(A2). Let
\[
\left\{  \left(  x_{1},\ldots,x_{n}\right)  \right\}  _{n}\in\left(
\mathbb{Q}_{p}^{d}\right)  ^{n}%
\]
be a sequence of configurations, and let $\left\{  \mu_{n}\right\}  _{n}$ be
associated empirical measures (defined by $\mu_{n}=\frac{1}{n}\sum_{i=1}%
^{n}\delta_{x_{i}}$). Assume that $\left\{  \frac{1}{n^{2}}H_{n}\left(
x_{1},\ldots,x_{n}\right)  \right\}  _{n}$ is a bounded sequence. Then the
sequence $\left\{  \mu_{n}\right\}  _{n}$ is tight, and as $n\rightarrow
+\infty$, it converges weakly in $\mathcal{P}(\mathbb{Q}_{p}^{d})$ (up to
extraction of a subsequence) to some probability measure $\mu$.
\end{lemma}

\begin{theorem}
\label{Theorem_2}Assume that $V$ is continuous and satisfies (A2). Assume that
for each $n$, $\left\{  \left(  x_{1},\ldots,x_{n}\right)  \right\}  _{n}$ is
a minimizer of $H_{n}$. Then
\[
\frac{1}{n}\sum_{i=1}^{n}\delta_{x_{i}}\rightarrow\mu_{0}\qquad\text{in the
weak sense of probability measures,}%
\]
where $\mu_{0}$ is the unique minimizer of $I$ as in Theorem \ref{Theorem_1},
and
\[
\lim_{n\rightarrow+\infty}\frac{1}{n^{2}}H_{n}\left(  x_{1},\ldots
,x_{n}\right)  =I(\mu_{0}).
\]

\end{theorem}

\begin{proof}
The proof follows from Lemma \ref{Lemma_6}, Proposition \ref{Prop_1}\ and
Theorems \ref{Theo_1A}, \ref{Theorem_1}, by using the reasoning given in
\cite{Serfaty-Book} for the proof of Theorem 2.2.
\end{proof}

\subsection{\label{Section_Spin_glass}Continuum limits of hierarchical models}

The energy function $-I\left(  \mu\right)  $ is the continuum limit of a
$p$-adic hierarchical Hamiltonian, which corresponds to a certain type of
$p$-adic hierarchical spin glass. A similar result was established by Lerner
and Missarov \cite[Theorem 2]{Lerner-Missarov}, see also \cite[Section
C]{Gubser et al}. For $L\in\mathbb{Z}$ fixed, we take the potential
\[
V\left(  x\right)  =
\begin{cases}
V_{0}\left(  x\right),  & \text{if }  x\in B_{L}^{d},\\
+\infty, & \text{if } x\notin B_{L}^{d},
\end{cases}
\]
where $V_{0}:B_{L}^{d}\rightarrow\mathbb{R}$ is a continuous function. Let
$\mathcal{P}(B_{L}^{d})$ denote the space of probability distributions
supported in $B_{L}^{d}$. Consider the functional%
\[
I_{L}(\mu):=\int_{B_{L}^{d}}\,\int_{B_{L}^{d}}g_{\alpha}\left(  x-y\right)
d\mu\left(  x\right)  d\mu\left(  y\right)  +\int_{B_{L}^{d}}V_{0}\left(
x\right)  d\mu\left(  x\right)  <\infty,
\]
for $\mu\in\mathcal{P}(B_{L}^{d})$. There exists a probability measure
$\mu_{0}\in\mathcal{P}(B_{L}^{d})$ such that $\min_{\mu\in\mathcal{P}%
(B_{L}^{d})}I_{L}(\mu)=I_{L}(\mu_{0})$. \ This fact follows from\ Theorem
\ref{Theorem_1} by noticing that the equilibrium measure must be supported in
$B_{L}^{d}$ due to the fact that the potential $V$ is infinite outside of this ball.

We set $G_{l}:=p^{-L}\mathbb{Z}_{p}^{d}/p^{l}\mathbb{Z}_{p}^{d}$, with
$l\geq-L$. \ By fixing an identification of $G_{l}$ with a subset of
$\mathbb{Q}_{p}^{d}$, $G_{l}$ becomes a finite ultrametric space, see e.g.
\cite[Section 3]{Zuniga-Nonlinearity}. We also pick $\rho:B_{L}^{d}%
\rightarrow\left[  0,\infty\right)  $ a continuous function. We now define the
following approximations of $\rho$ and $V_{0}$:%
\[
\rho_{l}\left(  x\right)  =\sum_{\widetilde{x}\in G_{l}}
\rho\left(  \widetilde{x}\right)  \Omega\left(  p^{l}\left\Vert x-\widetilde
{x}\right\Vert _{p}\right) \qquad \text{for }l\geq-L,
\]
and
\[
V_{l}\left(  x\right)  =
\sum_{\widetilde{x}\in G_{l}}
V_{0}\left(  \widetilde{x}\right)  \Omega\left(  p^{l}\left\Vert
x-\widetilde{x}\right\Vert _{p}\right)\qquad  \text{for }l\geq-L,
\]
which are test functions supported in $B_{L}^{d}$ satisfying that $\rho_{l}\overset{\left\Vert \cdot\right\Vert _\infty}{\rightarrow}\rho$ and
 $V_{l}\overset{\left\Vert \cdot\right\Vert_\infty}{\rightarrow}V_{0}$, see e.g. \cite[Lemma 1]{Zuniga-Nonlinearity}. Then
\[
I_{L}(\rho_{l}\,dx):=\sum_{\widetilde{x},\,\widetilde{y}\in G_{l}}
p^{-2ld}J_{\widetilde{x}\,\widetilde{y}}\rho\left(  \widetilde
{x}\right)  \rho\left(  \widetilde{y}\right)  +
\sum_{\widetilde{x}\in G_{l}}
p^{-ld}\rho\left(  \widetilde{x}\right)  V_{0}\left(  \widetilde{x}\right)  ,
\]
where%
\[
J_{\widetilde{x}\,\widetilde{y}}=
\begin{cases}
\left\Vert \widetilde{x}-\widetilde{y}\right\Vert _{p}^{\alpha-d}, & \text{if } \widetilde{x}\neq\widetilde{y},\\
\frac{p^{-l\left(  d+\alpha\right)  }\left(  1-p^{-d}\right)  }{1-p^{-\alpha}},
& \text{if } \widetilde{x}=\widetilde{y}.
\end{cases}
\]
The function $-I_{L}(\rho_{l}\,dx)$ is the Hamiltonian of a spin glass with $p$-adic
coupling, see \cite[Section C]{Gubser et al}. We now show that
\[
\lim_{l\rightarrow\infty}I_{L}(\rho_{l}\,dx)=I_{L}(\rho\, dx).
\]
Indeed, since $\rho_{l}\overset{\left\Vert \cdot\right\Vert _\infty}{\rightarrow}\rho$, $V_{l}\overset{\left\Vert \cdot\right\Vert_\infty}{\rightarrow}V_{0}$, there is a positive constant $C$ such that $\left\Vert
\rho_{l}\right\Vert _{\infty}<C\left\Vert \rho\right\Vert _{\infty}$ and
$\left\Vert V_{l}\right\Vert _{\infty}<C\left\Vert V_{0}\right\Vert _{\infty}$
for $l$ sufficiently large. Consequently by the dominated convergence
lemma,
\[
\int_{B_{L}^{d}}\,\int_{B_{L}^{d}}g_{\alpha}\left(  x-y\right)  \rho
_{l}\left(  x\right)  \rho_{l}\left(  y\right)  dxdy\rightarrow\int_{B_{L}%
^{d}}\,\int_{B_{L}^{d}}g_{\alpha}\left(  x-y\right)  \rho\left(  x\right)
\rho\left(  y\right)  dxdy,
\]
and
\[
\int_{B_{L}^{d}}V_{l}\left(  x\right)  \rho_{l}\left(  x\right)
dx\rightarrow\int_{B_{L}^{d}}V_{0}\left(  x\right)  \rho\left(  x\right)  dx,\qquad l\to\infty.
\]

\end{document}